%% file: main.tex
\documentclass[letterpaper, 10 pt, conference]{ieeeconf}  
\IEEEoverridecommandlockouts                              

  \let\labelindent\relax
\usepackage{amsmath} 
\usepackage{amssymb}  

\usepackage{tabulary}
\usepackage[english]{babel}
\usepackage{blindtext}
\usepackage[font=footnotesize]{caption}
\usepackage[makeroom]{cancel}
\usepackage{amsfonts}
\usepackage{subfig,pgfplots}
\usepackage{amsthm}  
\usepackage{amssymb}
\usepackage{graphicx}
\usepackage{pgf,tikz}
\pgfplotsset{ 
  compat=newest, 
   legend style =
  {font=\footnotesize },
  label style = {font=\footnotesize},
every tick label/.append style={font=\footnotesize}
  }
\usepackage{tabularx}
\usepackage{float}
\usepackage{cite}
\usepgfplotslibrary{fillbetween}
\usetikzlibrary{shapes.multipart}
\usepackage{enumitem} 
\usepackage{bbm} 

\usepackage[strict]{changepage}

\newcommand{\mat}[1]{\boldsymbol{#1}}

\renewcommand{\eqref}[1]{Eq.~(\ref{#1})}  

\newtheorem{remark}{Remark}
\newtheorem{theorem}{Theorem}
\newtheorem{lemma}{Lemma}

\newtheorem{proposition}{Proposition}

\newtheorem{problem}{Problem}
\newtheorem{example}{Example}

\def\h{1.3}
\def\hh{1.3}
\def\l{2.3}
\def\ll{1.9}

\title{\LARGE \bf
On Adaptive-Gain Control of Replicator Dynamics in Population Games
}

\author{Lorenzo Zino, Mengbin Ye, Alessandro Rizzo, and Giuseppe Carlo Calafiore
\thanks{L. Zino, A. Rizzo, and G.C. Calafiore are with the Department of Electronics and Telecommunications, Politecnico di Torino, Torino, Italy (\texttt{\{lorenzo.zino,alessandro.rizzo,giuseppe.calafiore\}@ polito.it}). 
M. Ye is with the Centre for Optimisation and Decision Science, Curtin University, Perth, Australia (\texttt{mengbin.ye@curtin.edu.au}).
A. Rizzo is also with the Institute for Invention, Innovation, and Entrepreneurship, New York University Tandon School of Engineering, Brooklyn NY, US. G.C. Calafiore is also with the College of Engineering \& Computer Science, VinUniversity, Hanoi, Vietnam. This work was partially supported by the Western Australian Government (Premier's Science Fellowship Program) and was carried out within the FAIR - Future Artificial Intelligence Research and received funding from the European Union Next-GenerationEU (PIANO NAZIONALE DI RIPRESA E RESILIENZA (PNRR) – MISSIONE 4 COMPONENTE 2, INVESTIMENTO 1.3 – D.D. 1555 11/10/2022, PE00000013). This manuscript reflects only the authors’ views and opinions, neither the European Union nor the European Commission can be considered responsible for them. }%
}

\begin{document}

\maketitle
\thispagestyle{empty}

\begin{abstract}
Controlling evolutionary game-theoretic dynamics is a problem of paramount importance for the systems and control community, with several applications spanning from social science to engineering. Here, we study a population of individuals who play a generic $2$-action matrix game, and whose actions evolve according to a replicator equation ---a nonlinear ordinary differential equation that captures salient features of the collective behavior of the population. Our objective is to steer such a population to a specified equilibrium that represents a desired collective behavior ---e.g., to promote cooperation in the prisoner's dilemma. To this aim, we devise an adaptive-gain controller, which regulates the system dynamics by adaptively changing the entries of the payoff matrix of the game. The adaptive-gain controller is tailored according to distinctive features of the game, and conditions to guarantee global convergence to the desired equilibrium are established. 
\end{abstract}

\section{Introduction}\label{sec:intro}

Evolutionary game theory is a mathematical framework that models the dynamics of  large populations, whose individuals repeatedly engage in strategic interactions~\cite{Sandholm2010,Hofbauer2009}. Such a framework, originally proposed for biological systems~\cite{Smith1973}, has found applications in a broad range of population dynamics, spanning from economics to social science and ecology, attracting a growing interest~\cite{Quijano2017,Barreiro-Gomez2016,Como2021,Stella2022cooperation}. In particular, controlling evolutionary game-theoretic dynamics has emerged as a problem of interest for the systems and control community~\cite{Quijano2017,riehl2018survey,Grammatico2017}. Notably, even steering the population to a desired equilibrium constitutes a challenge, since such an  equilibrium may be one among several, and may even be unstable or not globally attractive for the uncontrolled population dynamics. 

Diverse open-loop approaches have been proposed in the literature to tackle this problem. Several methods rely on the assumption that it is possible to act on the individuals directly ---e.g., by setting the behavior of part of the population who can act as a committed minority~\cite{centola2018experimental_tipping,Como2022targeting}--- or indirectly ---e.g., by incorporating features into their decision-making, such as sensitivity to trends~\cite{cdc2021,Zino2022nexus,zino2023ifac},  reciprocity~\cite{vanVeelen2012,Park2022cooperative}, or reputation~\cite{Giardini2021}. However, such  interventions on the individuals are not always feasible. To address this limitation,  alternate methods have been proposed, in which incentives are used to modify the structure of the payoff in order to favor the adoption of the desired action~\cite{riehl2018incentive,Eksin2020,Gong2022,Zhu2023}. However, these methods also suffer from   critical limitations. Namely, they require accurate information on the structure and characteristics of the game in order to design and tune the incentives in an open-loop fashion. This hinders the possibility to use them in many real-world scenarios, where uncertainty in the payoffs may cause open-loop approaches to fail, calling for the development of closed-loop schemes to control evolutionary game-theoretic dynamics.

In this paper, we fill in this gap by proposing a novel adaptive-gain controller for game-theoretic evolutionary dynamics~\cite{Smith1973,Sandholm2010}. In particular, we focus on the broad class of symmetric 2-player matrix games~\cite{riehl2018survey} and we consider a population where each individual plays the 2-player game against all the others, yielding a population game~\cite{Sandholm2010}. The emergent behavior of the action revision processes is captured at the population-level by a replicator equation~\cite{taylor1978replicator}, which is one of the most widely used and studied revision protocols in evolutionary game theory~\cite{Sandholm2010,Cressman2014}. In such a framework, the state of the population (i.e., the fraction of adopters of each action) is modeled by way of a nonlinear ordinary differential equation (ODE)~\cite{taylor1978replicator,Hofbauer2009,Sandholm2010}. For such dynamics, it has been shown that the population converges to a Nash equilibrium (NE) of the game~\cite{Sandholm2010}. For instance, in the well-known prisoner's dilemma, the replicator equation leads the entire population to the Pareto-inefficient equilibrium in which all players  defect~\cite{Cressman2014}. Adaptive-gain control~\cite{ioannou1996robust} has been used to stabilize complex network systems~\cite{yu2012distributed_adaptive,mei2016adaptive}, with applications to epidemics~\cite{Walsh2023_IFAC}, but to our knowledge, has not been used to control evolutionary dynamics.

Here, we encapsulate a feedback control scheme in the revision protocol to steer the population to a desired equilibrium. In particular, we consider  scenarios in which the desired equilibrium is either unstable or not globally attractive for the uncontrolled dynamics. To this aim, we propose a class of adaptive-gain controllers and we incorporate the adaptive gain within the payoff matrix. We tailor the adaptive-gain controllers to the three different classes of $2$-action matrix games ---coordination, dominant-strategy, and anti-coordination games~\cite{riehl2018survey,zino2017cdc} by suitably designing their adaptation functions, and we analytically study the corresponding nonlinear system of ODEs to assess their performance. In particular, we prove that an appropriately designed adaptive-gain controller is able i) lead the population to the desired equilibrium in coordination games, ii) promote  cooperation in social dilemmas, and iii) enforce consensus in anti-coordination games. Moreover, we show that in all these scenarios the gains converge to a constant value, which is key for the real-world applicability of our control scheme. Importantly, our adaptive-gain controller does not require full information on the payoff matrix to achieve its goal, addressing the limitations of other methods mentioned above.


\section{Model and Preliminaries}\label{sec:model}

We start by gathering the notational conventions used in this paper. The set of real and nonnegative real numbers are $\mathbb R$ and $\mathbb R_{+}$, respectively. Given two positive integers $n$ and $m$, a matrix $\mat{A}\in\mathbb R^{n\times m}$ is denoted with bold capital font, with $A_{ij}$ denoting the generic $j$th entry of its $i$th row.

\subsection{$2$-player matrix game}

We consider a population where each individual plays a symmetric 2-player matrix game~\cite{riehl2018survey} against all the others, yielding a population game~\cite{Sandholm2010}. We start by presenting the formalism of such matrix game. In this game, each player can choose between two actions, termed action $1$ and action $2$, characterized by a payoff matrix 
\begin{equation}\label{eq:payoff}
    {\mat A}=\begin{bmatrix}  a&b\\c&d \end{bmatrix}\in\mathbb R^{2\times 2}.
\end{equation}
In other words, a player would receive payoff equal to $a$ or $b$ for selecting action $1$ against an opponent who plays $1$ or $2$, respectively; while they would receive $c$ or $d$ for selecting action $2$ against an opponent who plays $1$ or $2$, respectively. 

Following a standard classification for 2-player matrix games~\cite{zino2017cdc}, the payoff matrix in \eqref{eq:payoff} determines three different classes of games, summarized in the following.
\begin{proposition}\label{prop:nash}\rm The payoff matrix in \eqref{eq:payoff} determines three classes of games:
\begin{enumerate}
    \item If $d>b$ and $a>c$, the game is a \emph{coordination game}; it has two pure NE\footnote{A Nash equilibrium (NE) is a configuration of actions in which no player can increase their payoff by unilaterally changing their action. A NE is said to be pure if each player chooses a strategy and mixed if at least one player chooses a probability distribution over the strategies. } $(1,1)$ and $(2,2)$, and a mixed NE, where action $1$ is played with probability equal to
    \begin{equation}\label{eq:mixed}
   x^*:=\frac{d-b}{a+d-b-c}. 
\end{equation}
    \item If $d>b$ and $a<c$, or $d<b$ and $a>c$, the game is a \emph{dominant-strategy game}, which has the unique (pure) NE: i) $(2,2)$ if $d>b$ and $a<c$; or ii) $(1,1)$ if $d<b$ and $a>c$.
    \item If $d<b$ and $a<c$, the game is an \emph{anti-coordination game}, which has two pure NE ($(1,2)$ and $(2,1)$), and a mixed NE where action $1$ is played with probability  $x^*$.
\end{enumerate}
\end{proposition}

To illustrate these three classes of games, we provide one example for each class, which will be used later in the paper. 

\begin{example}[Pure coordination game]\label{ex:pure}
A game with diagonal and strictly positive payoff matrix ($b=c=0$, $a,d>0$) is a particular type of coordination game, called {pure coordination game}, in which a player receives a (positive) payoff only if they coordinate with their opponent. 
\end{example}

\begin{example}[Prisoner's dilemma]\label{ex:pd}
A game with $c<a<d<b$ is a prisoner's dilemma in which action $1$ and $2$ represent defection and cooperation, respectively. A player receives $d$ for mutual cooperation, while $b>d$ captures the temptation to cheat (i.e, to defect if the other cooperates); $a$ is the punishment for mutual defection, which provides a smaller payoff than mutual cooperation ($d$), but larger than cooperating if the other defects ($c$). Prisoner's dilemma is a dominant-strategy game, with mutual defection (action $1$) as the unique NE. This is somewhat counter-intuitive, since the NE is not socially optimal: the players would receive a larger joint payoff if they both cooperate. 
\end{example}

\begin{example}[Minority game]\label{ex:min}
A minority game is a game with zero-diagonal and strictly positive off-diagonal payoff matrix ($a=d=0$, $b,c>0$). This is an anti-coordination game in which a player receives a (positive) payoff only if they do not coordinate with their opponent. 
\end{example}

\subsection{Evolutionary dynamics}

We consider a continuum of players, each one playing the 2-action game in \eqref{eq:payoff} against the entire population. Let us denote by $x(t)\in[0,1]$ the fraction of adopters of action $1$ at time $t\in\mathbb R$; consequently, the fraction of adopters of action $2$ is equal to $1-x(t)$. Then, following the standard approach of population games~\cite{Sandholm2010}, the total reward  associated with action $1$ and $2$ is given by the average payoff that a player receives for choosing action $1$ and $2$ from all games played, respectively, which is equal to
\begin{equation}
   \begin{bmatrix}  r_1(x,{\mat A})\\ r_2(x,{\mat A}) \end{bmatrix}={\mat A}\begin{bmatrix}  x\\1-x \end{bmatrix}= \begin{bmatrix}  ax+b(1-x)\\ cx+d(1-x)\end{bmatrix}.
\end{equation}
For the sake of readability, we have omitted to explicitly write the dependence of the variables on time, i.e., to use the notation $x(t)$. Throughout this paper, we adopt this convention except when we wish to  highlight the dependence.

In population games~\cite{Sandholm2010}, the players' actions revision is captured at the population-level by means of ODEs. Here, we use the replicator equation~\cite{taylor1978replicator}, yielding:
\begin{equation}\label{eq:replicator}
\begin{array}{lll}
    \dot x&=&x(1-x)(r_1(x,{\mat A})-r_2(x,{\mat A}))\\&=&x(1-x)((a+d-b-c)x+b-d),\end{array}
    \end{equation}
with initial condition $x(0)\in[0,1]$. Briefly, \eqref{eq:replicator} captures the tendency of players to imitate their peers who have higher payoff: the rate at which individuals switch action from $2$ to $1$ is proportional to the difference in the reward for playing $2$ with respect to the reward for playing $1$. 

The (uncontrolled) replicator equation in \eqref{eq:replicator} has been extensively studied in the literature~\cite{taylor1978replicator,Sandholm2010}, and we can fully characterize it through the following result. For its proof, omitted due to space constraints, we refer to~\cite{taylor1978replicator,Sandholm2010}. 
\begin{proposition}\label{prop:behavior}\rm
Consider the replicator equation in \eqref{eq:replicator}. If the payoff matrix ${\mat A}$ in \eqref{eq:payoff} is
\begin{enumerate}
    \item a coordination game ($d>b$, $a>c$), then $x(t)\to 0$ if $x(0)<x^*$, and $x(t)\to 1$ if $x(0)>x^*$;
    \item a dominant-strategy game, and
    \begin{enumerate}
        \item  
 $d>b$ and $a<c$, then $x(t)\to 0$ for any $x(0)<1$;
    \item $d<b$ and $a>c$, then $x(t)\to 1$ for any $x(0)>0$;
       \end{enumerate}
    \item  an anti-coordination game ($d<b$, $a<c$), then $x(t)\to x^*$ for any $x(0)\in(0,1)$,
\end{enumerate}
where $x^*$ is defined in \eqref{eq:mixed}.
\end{proposition}

\begin{remark}
The behavior of \eqref{eq:replicator} can be related to the Nash equilibria of the game and to the concept of evolutionary stability~\cite{Smith1973}. In fact, all Nash equilibria of \eqref{eq:payoff} are associated with equilibria of \eqref{eq:replicator}. However, \eqref{eq:replicator} may have some equilibria that are not Nash ---e.g., $x=1$ in scenario 2). Moreover, the equilibria that are asymptotically stable coincide with the evolutionarily stable strategies of the game. Hence, the replicator equation converges almost everywhere to an evolutionarily stable strategy~\cite{Smith1973,Sandholm2010}.
\end{remark}

Figure~\ref{fig:1} illustrates Proposition~\ref{prop:behavior}. In particular, it depicts the trajectories for the (uncontrolled) replicator equation in three representative examples of the classes of games described above. In the pure coordination game (Fig.~\ref{fig:1a}), the population converges to a pure configuration that depends on the initial condition; in the prisoner's dilemma, we observe convergence to the all-defector configuration (Fig.~\ref{fig:1b}); in the minority game (Fig.~\ref{fig:1c}), the population converges to the unique (mixed) NE.

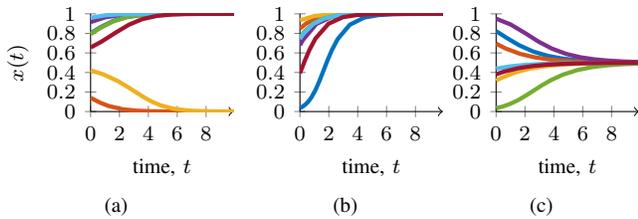
\begin{figure}
    \centering
    \subfloat[]{\input{1a}\label{fig:1a}}\,
    \subfloat[]{\input{1b}\label{fig:1b}}
    \subfloat[]{\input{1c}\label{fig:1c}}
\caption{Trajectories of the (uncontrolled) replicator equation in \eqref{eq:replicator} for (a) a pure coordination game with $a=d=1$; (b) a prisoner's dilemma with $c=0$, $a=1$, $d=2$, and $b=3$; and (c) a minority game with $b=c=1$. }
    \label{fig:1}
\end{figure}

\section{Problem Formalization}\label{sec:problem}

Proposition~\ref{prop:behavior} identifies two problems of interest. First, in scenario 1), the uncontrolled dynamics has two locally asymptotically stable equilibria ($x=0$ and $x=1$) and the initial condition determines the equilibrium reached by the system. An important problem is to design control strategies to steer the system to a desired equilibrium, regardless of the initial condition. In the context of social change (where action~$1$ and $2$ represent the status quo and the innovation, respectively), the problem translates into designing  intervention policies to incentivize social change when the innovators are initially a small minority~\cite{bass1969model}.

Second, in scenarios 2) or 3), the uncontrolled system has a unique globally asymptotically stable equilibrium, which is the unique NE of the game ---either a pure NE 2), or a mixed NE in 3). For these two scenarios, one problem of interest is to steer the system to a different equilibrium, e.g., to promote cooperation in the prisoner's dilemma (Example~\ref{ex:pd})~\cite{May1981,Stella2022cooperation} or reach coordination in an anti-coordination game (Example~\ref{ex:min})~\cite{ramazi2016networks}. 

Here, we consider both problems of interest. Without loss of generality, we focus on the case in which we aim to steer the system to the equilibrium $x=0$, i.e., all players collectively adopt action~$2$. Before formalizing our research problems, we want to point out that a naive intervention would be to permanently change the payoff matrix to shift the system to resemble scenario 2a) of Proposition~\ref{prop:behavior}. However, such a structural intervention has several limitations. First, in many real-world scenarios it may not be feasible, since permanent changes to the payoff structure may be economically unsustainable and unnecessary. In fact, in some scenarios (in particular, for coordination games), it may not be necessary to permanently change the payoff, since interventions might be needed only in a transient phase. In other scenarios, permanent changes could be necessary, but their implementation may not be required throughout the process, and can be limited to some specific conditions in which it is needed. Second, the aforementioned approaches would require  knowledge of the exact values in the payoff matrix. For these reasons, we focus here on designing more refined and dynamic control schemes.

Specifically, we assume that one or multiple entries of the payoff matrix ${\mat A}$ are controllable through the addition of a gain, i.e.,  that the payoff matrix can be written as
    \begin{equation}\label{eq:payoff_control}
    {\mat A}(t)= \mat{\hat A}+{\mat G}g(t)=\begin{bmatrix}   a & b\\ c& d \end{bmatrix}+\begin{bmatrix}   G_{11}& G_{12}\\G_{21}& G_{22} \end{bmatrix}g(t),
\end{equation}
where $\mat{\hat A}$ is the \emph{nominal payoff matrix} of the uncontrolled game from \eqref{eq:payoff}, $g(t):\mathbb R_+\to\mathbb R_+$ is a continuous function that quantifies the control \emph{gain}, and $\mat G\in\{0,1\}^{2\times2}$ is a $2\times 2$ matrix that determines which entries of the payoff function are controlled via the gain. We term $\mat G$ the \emph{control matrix}. In practical terms, the gains capture some temporary changes to the structure of the payoff matrix, which may be designed in order to achieve the desired equilibrium.     
    
We want to design control gains in an adaptive fashion, to obtain a closed-loop controller for the evolutionary dynamics that requires limited information on the game. Hence, we assume that the gain $g(t)$ is governed by the following ODE:
\begin{equation}\label{eq:gain}
  \dot g(t)=\phi(x(t))g(t),
\end{equation}
where the \emph{adaptation function} $\phi(x):[0,1]\to \mathbb R$ is continuously differentiable on its domain, with initial condition $g(0)>0$. Hence, an adaptive-gain controller can be characterized by the pair $(\mat{G},\phi)$, yielding the following planar system of coupled nonlinear ODEs:
\begin{equation}\label{eq:controlled_replicator}
\begin{array}{lll}  
    \dot x&=&x(1-x)\big((a+d-b-c)x+b-d\\&&+(G_{11}-G_{21})gx+(G_{12}-G_{22})g(1-x)\big),\\
    \dot g&=&\phi(x)g,
    \end{array}
\end{equation}
for which we can prove the following property.

\begin{lemma}\label{lemma:invariance}
The domain $\Omega:=[0,1]\times \mathbb R_+$ is positively invariant under \eqref{eq:controlled_replicator}.
\end{lemma}
\begin{proof}
The regularity of $\phi$ and the presence of the term $g$ on the right-hand side of \eqref{eq:gain} guarantees that $g(t)\in\mathbb R_+$ for all $t\in\mathbb R_+$~\cite{blanchini1999set}. 
\end{proof}
 
Besides guaranteeing convergence to the desired equilibrium $x=0$, it also makes sense for many practical applications to enforce that the gain converges to some constant value $\bar g$. 
An intuitive difference between the two problems of interest is that, for a coordination game the controller should be able to guide the system out of the basin of attraction of the undesired equilibrium and into the basin of the desired one; in the second problem the controller should enforce stability for an equilibrium that would otherwise be unstable. Consequently, while in the first problem it may be suitable to consider control policies that are only active in the transient, and for which $g(t)\to 0$, in the second scenario we may need the control to remain active in the steady state to ensure stability. Hence, we can now finally formulate our two research problems of interest.

\begin{problem}\label{pr:1}
    Design $(\mat{G},\phi)$ so that i) the system in \eqref{eq:replicator} for a coordination game converges to the desired equilibrium $x(t)\to 0$ and ii) the gain $g(t)\to 0$.
\end{problem}
    
\begin{problem}\label{pr:2}
    Design  $(\mat{G},\phi)$ so that i) the system in \eqref{eq:replicator} for a dominant-strategy game or an anti-coordination game converges to the desired equilibrium $x(t)\to 0$ and ii) the gain $g(t)\to \bar g$, for some constant $\bar g\in\mathbb R_+$.
\end{problem}


\section{Main Results}\label{sec:results}
\subsection{Coordination Games}\label{sec:coordination}

Let us consider a coordination game with payoff matrix from \eqref{eq:payoff_control} ($a>c$ and $d>b$). For the sake of readability, we define two constants $\alpha=a-c>0$ and $\beta=d-b>0$. Using this notation, the replicator equation in~\eqref{eq:replicator} reads 
\begin{equation}\label{eq:coordination}
    \dot x=x(1-x)\big((\alpha+\beta)x-\beta\big).
\end{equation}
As a consequence of Proposition~\ref{prop:behavior}, the trajectories of the uncontrolled dynamics in~\eqref{eq:coordination} converge to $x=0$ if $x(0)<x^*=\frac{\beta}{\alpha+\beta}$, and to $x=1$ if $x(0)>x^*$, with a saddle point equilibrium at $x^*$, as illustrated in Fig.~\ref{fig:1a}.

In order to solve Problem~\ref{pr:1} for any initial condition (especially when $x(0)>x^*)$, we observe that the gain dynamics should be such that $\dot g$ is strictly positive when $x$ is far from $0$, and negative in the proximity of $x=0$ so that $g(t)$ can converge to $0$. For this reason, we propose an adaptation function of the form 
\begin{equation}\label{eq:gain_dyn1}
    \phi^{(1)}(x)=k(x-h),
\end{equation}
with $k>0$ and $h>0$ that regulate the rate of adaptation and the rate of the decay of the gain, respectively.

Since our objective is to steer the population toward collective adoption of action $2$ and avoid reaching the other locally attractive equilibrium in which the entire population adopts action $1$, we will consider a control matrix that provides an adaptive advantage for choosing action $2$ against a player who plays action $1$, which can be interpreted as an incentive for innovators and early adopters of action~$2$. Hence, we adopt
\begin{equation}\label{eq:matrix1}
 \mat{G^{(1)}}=\begin{bmatrix}
    0&0\\1&0
    \end{bmatrix}.
\end{equation}


\begin{theorem}\label{th:1}\rm
The adaptive-gain control $(\mat{G^{(1)}},\phi^{(1)})$ solves Problem~\ref{pr:1} for any initial condition  $x(0)\in[0,1)$ and for any $k>0$ and $h\in(0,\frac{\beta}{\alpha+\beta})$.
\end{theorem}
\begin{proof}
Using~\eqref{eq:gain_dyn1} and \eqref{eq:matrix1}, \eqref{eq:controlled_replicator} for the coordination game reduces to the planar system
\begin{equation}\label{eq:controlled2}
    \begin{array}{lll}
            \dot x&=&x(1-x)\big((\alpha+\beta)x-\beta-gx\big)\\
            \dot g&=&kg(x-h).
    \end{array}
\end{equation}
which has three equilibria: $(0,0)$, which is (locally) asymptotically stable; $(1,0)$, which is an (unstable) saddle point; and $(\frac{\beta}{\alpha+\beta},0)$, which is an (unstable) source.

We now analyze \eqref{eq:controlled2}, to prove convergence to the desired equilibrium $(0,0)$ from any initial condition with $g(0)>0$ and $x(0)<1$. To this aim, we partition the positively invariant domain $\Omega$ of \eqref{eq:controlled2} into three regions $\mathcal A=[0,\frac{\beta}{\alpha+\beta})\times\mathbb R^+$, $\mathcal B=[\frac{\beta}{\alpha+\beta},1)\times\mathbb R^+$, and $\mathcal C=\{1\}\times\mathbb R^+$. In the following, we prove that any trajectory with initial conditions in $\mathcal A\cup \mathcal B$ converges to $(0,0)$. The proof is structured into two steps: i) we prove that any trajectory with initial conditions in $\mathcal A$ converges to $(0,0)$; then, ii) we  prove that any trajectory starting from $\mathcal B$ exits such set reaching $\mathcal A$.

In $\mathcal A$, we can bound the first equation in \eqref{eq:controlled2} as
 $\dot x\leq x(1-x)((\alpha+\beta)x-\beta)$, 
which is independent of $g$. Gronwall's inequality~\cite{Pachpatte_book} establishes that $x(t)\to 0$ for any  $x(0)\in \mathcal A$. As a consequence, we further have that $g(t)\to 0$, being $\dot g<0$ in the neighborhood of $x=0$. Hence, the basin of attraction of $(0,0)$ includes the entire set $\mathcal A$.

We focus now on the behavior of \eqref{eq:controlled2} in $\mathcal B$. Define $\mu:=k(\frac{\beta}{\alpha+\beta}-h)>0$, and observe from \eqref{eq:controlled2} that $\dot g\geq \mu g$. Hence, using Gronwall's inequality, we conclude that $g(t)\geq g(0)e^{\mu t}$, for any $t\geq 0$ provided that $x(t)\geq\frac{\beta}{\alpha+\beta}$. Inserting this bound into the first equation of \eqref{eq:controlled2} yields after some simplifications, 
$\dot x(t)\leq x(t)(1-x(t))(\alpha-\frac{\beta}{\alpha+\beta}g(0)e^{\mu t})$,
which holds for any $x\in[\frac{\beta}{\alpha+\beta},1)$ and is strictly negative for $t\geq\bar t=\frac{1}{\mu}\ln(\alpha)+\frac{1}{\mu}\ln(\alpha+\beta)-\frac{1}{\mu}\ln(\beta)-\frac{1}{\mu}\ln (g(0))$, which is a finite constant. This implies that $x(t)$ monotonically decreases for $t\geq \bar t$, until it is below $\frac{\beta}{\alpha+\beta}$ reaching $\mathcal A$. Once $\mathcal A$ is reached, the first part of this proof guarantees convergence to $(0,0)$, yielding the claim.
\end{proof}

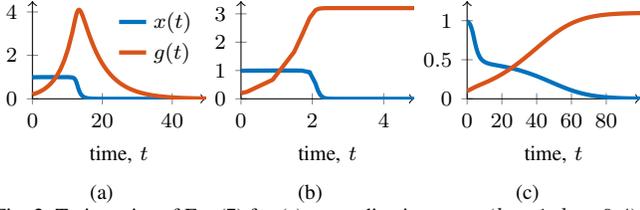
\begin{figure}
    \centering
    \subfloat[]{\input{2b}\label{fig:2b}}  
\subfloat[]{\input{3a}\label{fig:3a}}
    \subfloat[]{\input{3b}\label{fig:3b}}
\caption{Trajectories of \eqref{eq:controlled_replicator} for (a) a coordination game ($k=1$, $h=0.4$);  (b) a prisoner's dilemma ($k=2$, $h=1$); and (c) a minority game ($k=0.1$, $h=1$). Common parameters are $\alpha=\beta=1$ and $x(0)=0.99$.}
    \label{fig:2}
\end{figure}

Fig.~\ref{fig:2b} shows an exemplifying trajectory of the controlled replicator equation for a scenario in which the uncontrolled system converges to the NE $x=1$. Our controller, instead, is able to steer the system to the desired equilibrium $x = 0$. It is important to highlight that, even though Theorem~\ref{th:1} imposes a constraint on the parameter $h$ that depends on the entries of $\mat{\hat{A}}$, it is not necessary to know their precise values. In fact, one just needs a bound on the quantities $\alpha$ and $\beta$ (in particular, on their relative magnitude), in order to design the parameter $h$ to satisfy the condition in Theorem~\ref{th:1}. 


\subsection{Dominant-Strategy Games}\label{sec:dominant}

We consider a dominant-strategy game in which coordination on action $1$ is the NE ($a>c$ and $b>d$). The case in which action $2$ is the unique NE is trivial for Problem~\ref{pr:2}, since no control is needed to guarantee convergence to the desired equilibrium. The prisoner's dilemma illustrated in Example~\ref{ex:pd} is an example of this scenario. Similar to the coordination game, we define two strictly positive constants $\alpha=a-c>0$ and $\beta=-d+b>0$. 
Using this notation,~\eqref{eq:replicator} reduces to
\begin{equation}\label{eq:dom}
    \dot x=
    x(1-x)(\alpha+\beta(1-x)),
\end{equation}
and all trajectories with $x(0)>0$ converge to $x=1$ (Proposition~\ref{prop:behavior}), as can be observed in Fig.~\ref{fig:1b}.

In order to solve Problem~\ref{pr:2}, we want $\dot g$ to be positive when $x$ is far from $0$, and $g(t)$ should converge. The latter implies $\dot g\to 0$ by \eqref{eq:gain}. For this reason, we propose 
\begin{equation}\label{eq:gain_dyn2}
    \phi^{(2)}(x)=kx^h,
\end{equation}
with $k>0$ regulating the rate of adaptation and $h>0$ its sensitivity. Moreover, since our goal is to stabilize the equilibrium in which the entire population adopts action $2$, which would otherwise be unstable, we  use a  control matrix that promotes coordination on action $2$:
\begin{equation}\label{eq:matrix2}
 \mat{G^{(2)}}=\begin{bmatrix}
    0&0\\0&1
    \end{bmatrix}.
\end{equation}

\begin{theorem}\label{th:2}\rm
The adaptive-gain control $(\mat{G^{(2)}},\phi^{(2)})$ with $k>\alpha$ solves Problem~\ref{pr:2} for a dominant-strategy game, for any initial condition  $x(0)\in[0,1)$.
\end{theorem}
\begin{proof}
Using~(\ref{eq:gain_dyn2}--\ref{eq:matrix2}), the controlled replicator dynamics in \eqref{eq:controlled_replicator} for a dominant-strategy game with action~$1$ as the unique NE reduces to the following planar system of ODEs:
\begin{equation}\label{eq:controlled3}
    \begin{array}{lll}
            \dot x&=&x\big(1-x\big)\big(\alpha+(\beta-g)(1-x)\big)\\
            \dot g&=&kgx^h.
    \end{array}
\end{equation}

Being $x(t)\in[0,1]$ for all $t\in\mathbb R_+$ (Lemma~\ref{lemma:invariance}), then $\dot g\geq 0$, and thus $g(t)$ is monotonically nondecreasing. So either it converges to some constant $\bar g$ or it diverges to $+\infty$.

In the first scenario ($g(t)\to \bar g$), then necessarily $\dot g(t)\to 0$. This implies $x(t)\to 0$, solving Problem~\ref{pr:2}. In the following, we rule out the second scenario by way of contradiction. 

To obtain a contradiction, assume that $g(t)$ diverges to $+\infty$. First, we exclude the possibility that $x(t)$ converges to $1$. To this aim, we observe that, if $k>\alpha$, then $g(t)$ grows in a neighborhood of $x=1$ faster than an exponentially-growing function with exponent $\alpha$, while the exponential rate at which $x(t)\to 1$ cannot be larger than $\alpha$. Hence, the term $g(1-x)$ grows exponentially large, guaranteeing that there exists a time $\bar t$ such that $\dot x(t)<0$ for all $t\geq \bar t$, and thus $x(t)$ eventually decreases and cannot converge to $1$. Technical details are omitted due to space constraint. 

Since $x(t)$ does not converge to $1$ and we have assumed that $g(t)$ converges to $+\infty$, then we can find two positive constants $\varepsilon>0$ and $\tau_\varepsilon<\infty$ such that i) $x(\tau_\varepsilon)\leq 1-\varepsilon$ and ii) $g(\tau_\varepsilon)\geq \alpha/\varepsilon+\beta+\gamma$, where $\gamma>0$ is a small but positive constant. This implies that $\dot x(\tau_\varepsilon)\leq -\gamma x(1-x)^2$. By continuity, and since $g(t)$ is monotonically increasing, we can conclude that the bound $\dot x(t)\leq -\gamma \varepsilon^2 x$ holds true for any $t\geq\tau_\varepsilon$. Therefore, Gronwall's lemma yields
\begin{equation}\label{eq:exp_conv_bound}
    x(t)\leq (1-\varepsilon)e^{-\gamma\varepsilon^2(t-\tau_\varepsilon)},
\end{equation} for any $t\geq \tau_\varepsilon$. Finally, by integrating the second equation in \eqref{eq:controlled3}, and inserting \eqref{eq:exp_conv_bound} into it, we obtain
    \begin{equation}\begin{array}{l}
       \displaystyle \lim_{t\to\infty}g(t)=   \displaystyle\lim_{t\to\infty}g(0)e^{k\int_0^t x(s)^hds} \\\qquad= 
     g(0)e^{k\int_0^{\tau_\varepsilon}x(s)^h ds}e^{k\int_{\tau_\varepsilon}^{\infty}x(s)^h ds}\\\qquad\leq g(0)e^{k\int_0^{\tau_\varepsilon}1 ds}e^{k(1-\varepsilon)^h\int_{\tau_\varepsilon}^{\infty}(e^{-\gamma\varepsilon^2(t-\tau_\varepsilon)})^h ds}\\\qquad\leq g(0)e^{k\tau_\varepsilon}e^{k(1-\varepsilon)^h\int_0^\infty e^{-h\gamma\varepsilon^2s}ds}\\
     \qquad\leq g(0)e^{k\tau_\varepsilon}\exp\frac{k(1-\varepsilon)^h}{h\gamma\varepsilon^2}<+\infty,
    \end{array}\end{equation}
    which contradicts the assumption that $g(t)$ diverges to $+\infty$, yielding the claim.
\end{proof}



In Fig.~\ref{fig:3a} we illustrate the effectiveness of our approach, whose validity was proved in Theorem~\ref{th:2}. In this example, we illustrate how our adaptive-gain controller is able to promote cooperation for the prisoner's dilemma (see Example~\ref{ex:pd}), and the gain, after a steep increase, converges to a constant value, which guarantees stability of the desired equilibrium $x=0$.

\subsection{Anti-coordination game}\label{sec:anti}

We consider an anti-coordination game. Defining two strictly positive constants $\alpha=c-a>0$ and $\beta=b-d>0$, the replicator equation in~\eqref{eq:replicator} can be written as 
\begin{equation}\label{eq:anti}
    \dot x=x(1-x)(\beta-(\alpha+\beta)x),
\end{equation}
for which all the trajectories with $x(0)\in(0,1)$ converge to the mixed NE $x^*=\frac{\beta}{\alpha+\beta}$ (Proposition~\ref{prop:behavior}).

In order to solve Problem~\ref{pr:2}, we want $\dot g$ to be positive when $x$ is far from $0$, while $g(t)$ should eventually converge, implying $\dot g\to 0$. For this reason, we use a control matrix that promotes cooperation on action $2$ (i.e., control matrix $\mat{G^{(2)}}$), with the same adaptation function defined in~\eqref{eq:gain_dyn2}. 

\begin{theorem}\label{th:3}\rm
The adaptive-gain control $(\mat{G^{(2)}},\phi^{(2)})$ solves Problem~\ref{pr:2} for an anti-coordination game, for any initial condition  $x(0)\in[0,1)$.
\end{theorem}
\begin{proof}
The proof follows the one of Theorem~\ref{th:2}, i.e., after inserting~\eqref{eq:gain_dyn2} and~\eqref{eq:matrix2} into \eqref{eq:controlled_replicator}, we proceed by contradiction, showing that $g(t)$ cannot diverge, and thus $x(t)$ must converge to $0$. Here, we do not need any requirement on the velocity of the adaptation process, since $\dot x$ is always negative close to $x=1$, regardless of the value of $g(t)$.
\end{proof}

Figure~\ref{fig:3b} illustrates the effectiveness of our algorithm for a minority game (see Example~\ref{ex:min}), which is guaranteed by Theorem~\ref{th:3}. We observe that, after an initial transient in which the system quickly approaches the mixed NE $x^*$ of the uncontrolled system, the adaptation function adaptively increases the gain in order to leave such an equilibrium and enforce convergence to the desired consensus state. Importantly, we observe that no information on the payoff structure is needed (not even estimates of $\alpha$ and $\beta$), as the adaptive-gain mechanisms is able to autonomously change the value of the gain in order to reach our goal.


\section{Conclusion}\label{sec:conclusion}

In this paper, we proposed and analyzed a novel adaptive-gain controller for evolutionary-game dynamics. In particular, we focused on a population of individuals who play a $2$-action matrix game and revise their actions following a replicator equation. In this scenario, we illustrated how to design the controller  to steer the system to a desired equilibrium, in terms of deciding which entry of the payoff matrix to control and how to design the adaptation function depending on the characteristics of the game, and we establish analytical guarantees on its effectiveness. In view of its close-loop structure and on the limited amount of information on the game needed to design it, our controller may have important applications in designing interventions to promote a desired collective behavior, such as cooperation in social dilemmas.

The promising results in this paper pave the way for several directions of future work. First, our results suggests that, while always able to solve the problem of interest, the efficiency of the adaptive-gain control schemes may vary depending on the adaptation function. Here, we analyzed three specific implementations, while a general treatment of the adaptive-gain controller and the study of  optimal control policies is left for future research. Second, it will be of interest to consider evolutionary dynamics in networked and structured populations~\cite{Barreiro-Gomez2016,Como2021,Govaert2022network}, where each population may have access only to local information to update their gain, as well as to consider more complex classes of revision protocols and games~\cite{Sandholm2010}. Third, in many real-world scenarios the  population behavior affects some boundary conditions of the system, ultimately impacting the game's payoff matrix. This is captured by population dynamics with  environmental feedback~\cite{Gong2022}, for which our controllers should be extended.


\end{document}

%% file: 1a.tex
%
%
\definecolor{mycolor1}{rgb}{0.00000,0.44700,0.74100}%
\definecolor{mycolor2}{rgb}{0.85000,0.32500,0.09800}%
\definecolor{mycolor3}{rgb}{0.92900,0.69400,0.12500}%
\definecolor{mycolor4}{rgb}{0.49400,0.18400,0.55600}%
\definecolor{mycolor5}{rgb}{0.46600,0.67400,0.18800}%
\definecolor{mycolor6}{rgb}{0.30100,0.74500,0.93300}%
\definecolor{mycolor7}{rgb}{0.63500,0.07800,0.18400}%
\begin{tikzpicture}

\begin{axis}[%
 axis lines=left,
 x   axis line style={->},
  y   axis line style={-},
  width=\ll cm,
height=\hh cm,
at={(0.758in,0.481in)},
scale only axis,
xmin=0,
xmax=9.9,
ylabel={$x(t)$},
xlabel={time, $t$},
ymin=0,
ymax=1,
axis background/.style={fill=white},legend style={at={(1,1)}, legend image post style={xscale=0.6},legend cell align=left, align=left, draw=none,fill=none},
]
\addplot [color=mycolor1, ultra thick, forget plot]
  table[row sep=crcr]{%
0	0.8002804688888\\
0.66697931935524	0.861812581749114\\
1.54205847824255	0.925911871575676\\
2.54205847824255	0.969225832504406\\
3.54205847824255	0.988657366520282\\
4.54205847824255	0.996068553612531\\
5.54205847824255	0.998671258910478\\
6.54205847824255	0.999554993436979\\
7.54205847824255	0.999851429457981\\
8.54205847824255	0.999950450279301\\
9.54205847824255	0.999983480511108\\
10	0.999989577526743\\
};

\addplot [color=mycolor2, ultra thick, forget plot]
  table[row sep=crcr]{%
0	0.141886338627215\\
0.130165019080834	0.130743837424712\\
0.650352121397383	0.0911136788646699\\
1.23757028694723	0.057186239744598\\
2.23757028694723	0.0227224458368937\\
3.23757028694723	0.00817047420316143\\
4.23757028694723	0.00280192621966711\\
5.23757028694723	0.000943263512602652\\
6.23757028694723	0.000315476408343737\\
7.23757028694723	0.000105276939022277\\
8.23757028694722	3.51054724829208e-05\\
9.23757028694722	1.17032875599185e-05\\
10	5.31755403165288e-06\\
};

\addplot [color=mycolor3, ultra thick, forget plot]
  table[row sep=crcr]{%
0	0.421761282626275\\
0.884160485395348	0.380418593606555\\
1.88416048539535	0.311981918442912\\
2.61290102750154	0.24778154820227\\
3.19107259852911	0.192470837001101\\
3.76924416955669	0.13932457837781\\
4.30156916954676	0.0972448080446552\\
4.86920236737188	0.0625600641383618\\
5.75368522176967	0.0284470198221932\\
6.69923080973803	0.0111361776615062\\
6.95139327272919	0.00871717701339644\\
7.20355573572035	0.00681260538706116\\
7.5418776674749	0.00488263036397817\\
7.88918416760253	0.00346233032774615\\
8.23398176348226	0.00245823621503149\\
8.57547817517893	0.00174956699456883\\
8.91442132578745	0.00124761623724756\\
9.25154446255431	0.000890914261843869\\
9.58738628944761	0.000636819102338727\\
9.93547182500404	0.00044951157732864\\
10	0.000421456652158362\\
};

\addplot [color=mycolor4, ultra thick, forget plot]
  table[row sep=crcr]{%
0	0.915735525189067\\
1	0.964088829220153\\
2	0.986560155159742\\
3	0.995309192709405\\
4	0.99841042959887\\
5	0.999467149076774\\
6	0.99982204609989\\
7	0.999940644437036\\
8	0.999980210629019\\
9	0.999993403077971\\
10	0.999997800974311\\
};

\addplot [color=mycolor5, ultra thick, forget plot]
  table[row sep=crcr]{%
0	0.792207329559554\\
0.658777354048206	0.853764637769752\\
1.52380256646239	0.919709813294973\\
2.52380256646239	0.966121044610316\\
3.52380256646239	0.987396599645131\\
4.5238025664624	0.995613250096272\\
5.5238025664624	0.998515030476591\\
6.5238025664624	0.999502396760794\\
7.5238025664624	0.999833838414475\\
8.5238025664624	0.999944580025669\\
9.5238025664624	0.999981523028239\\
10	0.999988559043964\\
};

\addplot [color=mycolor6, ultra thick, forget plot]
  table[row sep=crcr]{%
0	0.959492426392903\\
1	0.984635754805855\\
2	0.99460364834718\\
3	0.998166880278115\\
4	0.999384979467399\\
5	0.999794544354507\\
6	0.999931464671913\\
7	0.999977149313356\\
8	0.999992382484413\\
9	0.999997460759232\\
10	0.999999153578754\\
};

\addplot [color=mycolor7, ultra thick, forget plot]
  table[row sep=crcr]{%
0	0.655740699156587\\
0.746057308083157	0.714872893970097\\
1.74605730808316	0.808672507520108\\
2.64976981795966	0.888368874233886\\
3.56865407630857	0.945367724460069\\
4.56865407630857	0.978442211340199\\
5.56865407630857	0.992276617785328\\
6.56865407630857	0.997355414082496\\
7.56865407630857	0.999110195386989\\
8.56865407630857	0.999702459378267\\
9.56865407630857	0.999900714704939\\
10	0.999935625616407\\
};

\end{axis}
\end{tikzpicture}%

%% file: 1b.tex
%
%
\definecolor{mycolor1}{rgb}{0.00000,0.44700,0.74100}%
\definecolor{mycolor2}{rgb}{0.85000,0.32500,0.09800}%
\definecolor{mycolor3}{rgb}{0.92900,0.69400,0.12500}%
\definecolor{mycolor4}{rgb}{0.49400,0.18400,0.55600}%
\definecolor{mycolor5}{rgb}{0.46600,0.67400,0.18800}%
\definecolor{mycolor6}{rgb}{0.30100,0.74500,0.93300}%
\definecolor{mycolor7}{rgb}{0.63500,0.07800,0.18400}%
\begin{tikzpicture}

\begin{axis}[%
 axis lines=left,
 x   axis line style={->},
  y   axis line style={-},
  width=\ll cm,
height=\hh cm,
at={(0.758in,0.481in)},
scale only axis,
xmin=0,
xmax=9.9,
xlabel={time, $t$},
ymin=0,
ymax=1,
]
\addplot [color=mycolor1, ultra thick, forget plot]
  table[row sep=crcr]{%
0	0.0357116785741896\\
0.0422355199853334	0.0386789430015311\\
0.209750101726733	0.0528570191584411\\
0.384833485627049	0.072637502539853\\
0.573654528949673	0.101035670619203\\
0.787292010402924	0.143720026654592\\
1.06003222868086	0.216148166468395\\
1.34550658738262	0.311657628676302\\
1.63098094608437	0.41929682209038\\
1.9019913133276	0.522149928439287\\
2.20922820573686	0.628547899693076\\
2.6136465634553	0.742141366996517\\
3.45045538754255	0.888210264598034\\
4.45045538754255	0.962277227300842\\
5.45045538754255	0.987408000127283\\
6.45045538754255	0.995802006432147\\
7.45045538754255	0.998600644342944\\
8.45045538754255	0.999533547208047\\
9.45045538754255	0.99984451570245\\
10	0.999910783936159\\
};

\addplot [color=mycolor2, ultra thick, forget plot]
  table[row sep=crcr]{%
0	0.849129305868777\\
0.460742810594324	0.904388921495328\\
1.15981626090621	0.953100377320353\\
2.11729812665396	0.983336945414606\\
3.11729812665396	0.994444118501388\\
4.11729812665395	0.998147982781865\\
5.11729812665395	0.999382658826381\\
6.11729812665395	0.999794219530981\\
7.11729812665395	0.999931406507445\\
8.11729812665395	0.999977135502375\\
9.11729812665395	0.999992378500788\\
10	0.999997010456432\\
};

\addplot [color=mycolor3, ultra thick, forget plot]
  table[row sep=crcr]{%
0	0.933993247757551\\
1	0.97790281913543\\
2	0.992630705337654\\
3	0.997543436091048\\
4	0.999181140460754\\
5	0.999727046638658\\
6	0.999909015539494\\
7	0.999969671846249\\
8	0.999989890615407\\
9	0.999996630205135\\
10	0.999998876735045\\
};

\addplot [color=mycolor4, ultra thick, forget plot]
  table[row sep=crcr]{%
0	0.678735154857773\\
0.188467687695932	0.729505137437449\\
1.13080612617559	0.896096283870895\\
2.13080612617559	0.964996181132461\\
3.13080612617559	0.98831788340655\\
4.13080612617559	0.996105433892472\\
5.13080612617559	0.998701791760988\\
6.13080612617559	0.999567263196718\\
7.13080612617559	0.999855754372105\\
8.13080612617559	0.999951918123043\\
9.13080612617559	0.999983972707644\\
10	0.999993603356688\\
};

\addplot [color=mycolor5, ultra thick, forget plot]
  table[row sep=crcr]{%
0	0.757740130578333\\
0.265825140692533	0.812029177355859\\
0.901709159045406	0.900085548737577\\
1.6698635748007	0.954652591998061\\
2.6698635748007	0.984853384236044\\
3.6698635748007	0.994949978932695\\
4.6698635748007	0.998316617050356\\
5.6698635748007	0.999438870772432\\
6.6698635748007	0.999812956865711\\
7.6698635748007	0.999937652286406\\
8.6698635748007	0.999979217428722\\
9.6698635748007	0.999993072476238\\
10	0.99999502353244\\
};

\addplot [color=mycolor6, ultra thick, forget plot]
  table[row sep=crcr]{%
0	0.743132468124916\\
0.247794274131303	0.796894399031665\\
0.90001504237442	0.893611241552819\\
1.67581069523124	0.952103662243124\\
2.67581069523124	0.983998251342196\\
3.67581069523124	0.994664728823485\\
4.67581069523124	0.998221526047578\\
5.67581069523124	0.999407173488758\\
6.67581069523124	0.999802391094014\\
7.67581069523124	0.999934130362119\\
8.67581069523124	0.999978043453945\\
9.67581069523124	0.999992681151312\\
10	0.99999471080417\\
};

\addplot [color=mycolor7, ultra thick, forget plot]
  table[row sep=crcr]{%
0	0.392227019534168\\
0.081869826921455	0.423695768636982\\
0.3530000466254	0.526400684736166\\
0.662721595147253	0.63302421708386\\
1.07363883916741	0.74705673625535\\
1.97740047224736	0.898819848378565\\
2.97740047224736	0.96593223585934\\
3.97740047224736	0.988631008473108\\
4.97740047224736	0.99620985018322\\
5.97740047224736	0.998736598720844\\
6.97740047224736	0.999578865573327\\
7.97740047224736	0.999859621833073\\
8.97740047224736	0.999953207276776\\
10	0.999984931262298\\
};

\end{axis}
\end{tikzpicture}%

%% file: 1c.tex
%
%
\definecolor{mycolor1}{rgb}{0.00000,0.44700,0.74100}%
\definecolor{mycolor2}{rgb}{0.85000,0.32500,0.09800}%
\definecolor{mycolor3}{rgb}{0.92900,0.69400,0.12500}%
\definecolor{mycolor4}{rgb}{0.49400,0.18400,0.55600}%
\definecolor{mycolor5}{rgb}{0.46600,0.67400,0.18800}%
\definecolor{mycolor6}{rgb}{0.30100,0.74500,0.93300}%
\definecolor{mycolor7}{rgb}{0.63500,0.07800,0.18400}%
\begin{tikzpicture}

\begin{axis}[%
 axis lines=left,
 x   axis line style={->},
  y   axis line style={-},
  width=\ll cm,
height=\hh cm,
at={(0.758in,0.481in)},
scale only axis,
xmin=0,
xmax=9.9,
xlabel={time, $t$},
ymin=0,
ymax=1,
axis background/.style={fill=white},legend style={at={(1,1)}, legend image post style={xscale=0.6},legend cell align=left, align=left, draw=none,fill=none},
]
\addplot [color=mycolor1, ultra thick, forget plot]
  table[row sep=crcr]{%
0	0.823457828327293\\
0.700476968423483	0.756526675227171\\
1.56572219338037	0.680783592155835\\
2.56572219338037	0.614344461419052\\
3.56572219338037	0.570345724556662\\
4.56572219338037	0.542791711385093\\
5.56572219338037	0.525918634964957\\
6.56572219338037	0.515673668726766\\
7.56572219338037	0.509472711816514\\
8.56572219338037	0.505723803818136\\
9.56572219338037	0.503458287458194\\
10	0.502783009252472\\
};

\addplot [color=mycolor2, ultra thick, forget plot]
  table[row sep=crcr]{%
0	0.694828622975817\\
0.672765067195254	0.644638210738461\\
1.67276506719525	0.589966546881982\\
2.67276506719525	0.55496610543644\\
3.67276506719525	0.533347288662102\\
4.67276506719525	0.52017820391599\\
5.67276506719525	0.512197835034625\\
6.67276506719525	0.507371035176405\\
7.67276506719525	0.504453666890029\\
8.67276506719525	0.502690830560363\\
9.67276506719525	0.501625726336177\\
10	0.501380305227894\\
};

\addplot [color=mycolor3, ultra thick, forget plot]
  table[row sep=crcr]{%
0	0.317099480060861\\
0.320248871383846	0.341224592098179\\
1.32024887138385	0.40063961052743\\
2.32024887138385	0.439144524890608\\
3.32024887138385	0.46304490986194\\
4.32024887138385	0.477630923434184\\
5.32024887138385	0.486476034279999\\
6.32024887138385	0.491827212868633\\
7.32024887138385	0.495061820338059\\
8.32024887138385	0.497016416287275\\
9.32024887138385	0.498197396081761\\
10	0.498717733938833\\
};

\addplot [color=mycolor4, ultra thick, forget plot]
  table[row sep=crcr]{%
0	0.950222048838355\\
1	0.891525913138387\\
2	0.803550902888398\\
2.94061682739422	0.715435128912546\\
3.85553410178442	0.644580365033057\\
4.85553410178442	0.58992844687108\\
5.85553410178442	0.554942305413803\\
6.85553410178442	0.533332729110986\\
7.85553410178442	0.520169367053044\\
8.85553410178442	0.512192487115442\\
9.85553410178442	0.507367802156662\\
10	0.506854462519949\\
};

\addplot [color=mycolor5, ultra thick, forget plot]
  table[row sep=crcr]{%
0	0.0344460805029088\\
0.0889843173105263	0.0373015629456747\\
0.482068726110895	0.0525041152917884\\
0.916736691650515	0.0749184526320879\\
1.48767523729975	0.114074942078589\\
2.04597965620979	0.161803029224255\\
2.60428407511984	0.214738813171397\\
3.14955169108196	0.266227724272951\\
3.77622304274123	0.319691713358694\\
4.61762599634628	0.377038995955164\\
5.61762599634628	0.424136987657893\\
6.61762599634628	0.453800563392283\\
7.61762599634628	0.472005624812984\\
8.61762599634628	0.48306846743283\\
9.61762599634628	0.489766493735941\\
10	0.491547331738964\\
};

\addplot [color=mycolor6, ultra thick, forget plot]
  table[row sep=crcr]{%
0	0.438744359656398\\
1	0.462799394797505\\
2	0.477481747015127\\
3	0.486385719176614\\
4	0.491772605984296\\
5	0.495028819515017\\
6	0.496996476268349\\
7	0.498185348541088\\
8	0.498903643106484\\
9	0.499337616614021\\
10	0.499599809795902\\
};

\addplot [color=mycolor7, ultra thick, forget plot]
  table[row sep=crcr]{%
0	0.381558457093008\\
0.546081268979408	0.408790886554257\\
1.54608126897941	0.444257327090737\\
2.54608126897941	0.466177563235279\\
3.54608126897941	0.479533387257626\\
4.54608126897941	0.487627620504963\\
5.54608126897941	0.492523445124281\\
6.54608126897941	0.495482567111532\\
7.54608126897941	0.497270640947238\\
8.54608126897941	0.498350995326847\\
9.54608126897941	0.499003722613558\\
10	0.499206118332234\\
};

\end{axis}
\end{tikzpicture}%

%% file: 2b.tex
%
%
\definecolor{mycolor1}{rgb}{0.00000,0.44700,0.74100}%
\definecolor{mycolor2}{rgb}{0.85000,0.32500,0.09800}%
\definecolor{mycolor3}{rgb}{0.92900,0.69400,0.12500}%
\definecolor{mycolor4}{rgb}{0.49400,0.18400,0.55600}%
\definecolor{mycolor5}{rgb}{0.46600,0.67400,0.18800}%
\definecolor{mycolor6}{rgb}{0.30100,0.74500,0.93300}%
\definecolor{mycolor7}{rgb}{0.63500,0.07800,0.18400}%
\begin{tikzpicture}

\begin{axis}[%
 axis lines=left,
 x   axis line style={->},
  y   axis line style={->},
  width=\l cm,
height=\h cm,
at={(0.758in,0.481in)},
scale only axis,
xmin=0,
xmax=49.5,
ylabel={},
xlabel={time, $t$},
ymin=0,
ymax=4.5,
axis background/.style={fill=white},legend style={at={(1,1)}, legend image post style={xscale=0.6},legend cell align=left, align=left, draw=none,fill=none},
]
\addplot [ultra thick, color=mycolor1]
  table[row sep=crcr]{%
0	0.99\\
0.212871731483863	0.991514467109075\\
0.425743462967725	0.992788947401566\\
0.638615194451588	0.993860814089521\\
0.851486925935451	0.99476307192064\\
1.52115343741342	0.996829582312619\\
2.19081994889138	0.998049259774778\\
2.86048646036935	0.998725143110301\\
3.53015297184732	0.999112930561486\\
4.38230018797808	0.999415402532342\\
5.23444740410884	0.999563938870549\\
6.0865946202396	0.999617660463883\\
6.93874183637036	0.999618053755602\\
8.03056108752861	0.997866930839312\\
9.12238033868685	0.998679724510216\\
10.2141995898451	0.999835020726312\\
11.3060188410033	0.98410355527305\\
11.4029345742904	0.98076192183414\\
11.4998503075774	0.976596215907322\\
11.5967660408645	0.971385952523667\\
11.6936817741515	0.964853160597856\\
11.7905975074386	0.956657659543915\\
11.8875132407256	0.946376239210301\\
11.9844289740127	0.933511620723152\\
12.0813447072997	0.917497660071395\\
12.2860029138599	0.87120278490336\\
12.49066112042	0.802112018769409\\
12.6953193269802	0.7094758153272\\
12.8999775335403	0.602265546426356\\
13.0844308203606	0.505221988304573\\
13.2688841071808	0.415584626426779\\
13.4533373940011	0.338130489348196\\
13.6377906808213	0.274251374422543\\
13.8242663814392	0.222025143716654\\
14.0107420820572	0.180195451327657\\
14.1972177826751	0.146819801326737\\
14.3836934832931	0.120084284085148\\
14.5641795682874	0.0990947552037662\\
14.7446656532818	0.0819959257471852\\
14.9251517382761	0.0680263099303905\\
15.1056378232705	0.0565393776002114\\
15.2893916347163	0.0468788325062862\\
15.473145446162	0.0389158086956812\\
15.6568992576078	0.0323451339853716\\
15.8406530690536	0.0269024815465198\\
16.0330113024871	0.0221840182878504\\
16.2253695359207	0.0183007008954721\\
16.4177277693542	0.0151069260654508\\
16.6100860027878	0.0124733395146024\\
16.8063083626923	0.0102547525907304\\
17.0025307225968	0.00843146639494907\\
17.1987530825013	0.00693540402449821\\
17.3949754424058	0.00570521096555086\\
17.5922913794992	0.00468528517469573\\
17.7896073165925	0.00384767315345303\\
17.9869232536858	0.00316107504501218\\
18.1842391907792	0.00259712359716695\\
18.381560097627	0.00213241395995979\\
18.5788810044748	0.00175084278085424\\
18.7762019113226	0.00143814274580743\\
18.9735228181704	0.0011813623760073\\
19.1706124757053	0.000970042099328846\\
19.3677021332403	0.000796528526245166\\
19.5647917907753	0.000654332479700717\\
19.7618814483103	0.000537561628113749\\
19.9587678300017	0.000441446020875787\\
20.155654211693	0.000362523471664595\\
20.3525405933844	0.000297842268449697\\
20.5494269750757	0.000244722366161174\\
20.7721874308289	0.00019577619706026\\
20.9949478865821	0.000156627020904805\\
21.2177083423353	0.000125427944673175\\
21.4404687980885	0.00010046811512492\\
21.6991691604157	7.74835007007497e-05\\
21.9578695227428	5.97611675734637e-05\\
22.21656988507	4.6207900652206e-05\\
22.4752702473971	3.57565460803523e-05\\
22.7812208088007	2.62531973239828e-05\\
23.0871713702043	1.92746277447683e-05\\
23.3931219316079	1.42565639926381e-05\\
23.6990724930115	1.05753741705391e-05\\
24.0703764631503	7.22435977916374e-06\\
24.4416804332892	4.92856426375454e-06\\
24.8129844034281	3.45548690164033e-06\\
25.1842883735669	2.45449545393447e-06\\
25.6498906433014	1.47927647184372e-06\\
26.1154929130359	8.78715930605033e-07\\
26.5810951827703	5.99858002799513e-07\\
27.0466974525048	4.4019519631427e-07\\
27.6555541649569	1.84905515767629e-07\\
28.2644108774089	5.42108434049077e-08\\
28.8732675898609	7.24001091537742e-08\\
29.482124302313	9.78607772648816e-08\\
30.3119878778925	-3.77394845771394e-08\\
31.1418514534721	-8.82207214756675e-08\\
31.9717150290517	2.55092966511704e-08\\
32.8015786046313	1.00204303493375e-07\\
33.870217925904	-3.65902827928099e-07\\
34.9388572471766	-4.9539335106151e-07\\
36.0074965684493	1.53830145025173e-07\\
37.076135889722	5.03995459310513e-07\\
37.8184919701046	3.14349596740727e-08\\
38.5608480504873	-1.67984715018307e-07\\
39.3032041308699	8.504930171187e-08\\
40.0455602112525	2.68984241704303e-07\\
40.7879162916351	1.67798029550079e-08\\
41.5302723720177	-8.96496380865786e-08\\
42.2726284524003	4.53920535718819e-08\\
43.0149845327829	1.43555828191515e-07\\
43.9041183472765	-1.22025031525568e-07\\
44.79325216177	-2.13302652337607e-07\\
45.6823859762636	5.69195044449267e-08\\
46.5715197907572	2.25064593319827e-07\\
47.4286398430679	-1.30456927738161e-07\\
48.2857598953786	-2.57958434429922e-07\\
49.1428799476893	7.06295399563805e-08\\
50	2.80971958471625e-07\\
};
\addlegendentry{$x(t)$}

\addplot [ultra thick,color=mycolor2]
  table[row sep=crcr]{%
0	0.2\\
0.212871731483863	0.21031811000309\\
0.425743462967725	0.221194740094067\\
0.638615194451588	0.232657125528666\\
0.851486925935451	0.244734032169363\\
1.52115343741342	0.287071737331252\\
2.19081994889138	0.336879566250886\\
2.86048646036935	0.395459631219129\\
3.53015297184732	0.464299439023306\\
4.38230018797808	0.569563998222988\\
5.23444740410884	0.698685810592018\\
6.0865946202396	0.857074065847756\\
6.93874183637036	1.05143007421589\\
8.03056108752861	1.36520134341458\\
9.12238033868685	1.77470509151393\\
10.2141995898451	2.30690366809745\\
11.3060188410033	2.98697815026591\\
11.4029345742904	3.05519372447152\\
11.4998503075774	3.12451610606663\\
11.5967660408645	3.1948337142704\\
11.6936817741515	3.26599196153902\\
11.7905975074386	3.33778480694861\\
11.8875132407256	3.40994209128053\\
11.9844289740127	3.48210381979973\\
12.0813447072997	3.55380736486853\\
12.2860029138599	3.70163397007005\\
12.49066112042	3.8368909520466\\
12.6953193269802	3.95020051170617\\
12.8999775335403	4.03391199110503\\
13.0844308203606	4.07962029672424\\
13.2688841071808	4.0975795474877\\
13.4533373940011	4.09044308161832\\
13.6377906808213	4.06193433552646\\
13.8242663814392	4.01597590800202\\
14.0107420820572	3.95669571786587\\
14.1972177826751	3.88732909633742\\
14.3836934832931	3.81063265170381\\
14.5641795682874	3.73150567008956\\
14.7446656532818	3.6490113095934\\
14.9251517382761	3.56431696828871\\
15.1056378232705	3.47839990149111\\
15.2893916347163	3.39045746580305\\
15.473145446162	3.30260104011559\\
15.6568992576078	3.2152923761258\\
15.8406530690536	3.12891427548702\\
16.0330113024871	3.03980525732849\\
16.2253695359207	2.95225756134038\\
16.4177277693542	2.86644469436213\\
16.6100860027878	2.78250416038561\\
16.8063083626923	2.69890440946608\\
17.0025307225968	2.61740141234231\\
17.1987530825013	2.53802624827987\\
17.3949754424058	2.46079572478199\\
17.5922913794992	2.38529651971561\\
17.7896073165925	2.31194429407848\\
17.9869232536858	2.2407118328122\\
18.1842391907792	2.17156724945972\\
18.381560097627	2.10447087145873\\
18.5788810044748	2.03937954233136\\
18.7762019113226	1.97624688594781\\
18.9735228181704	1.91502570797861\\
19.1706124757053	1.85573589826393\\
19.3677021332403	1.79825446307677\\
19.5647917907753	1.74253166448746\\
19.7618814483103	1.68851837253399\\
19.9587678300017	1.63621898549838\\
20.155654211693	1.58552858144121\\
20.3525405933844	1.53639983715068\\
20.5494269750757	1.48878649570854\\
20.7721874308289	1.43668624973248\\
20.9949478865821	1.38640383311117\\
21.2177083423353	1.33787702967256\\
21.4404687980885	1.29104553440793\\
21.6991691604157	1.23870870182419\\
21.9578695227428	1.18849104263171\\
22.21656988507	1.14030740449552\\
22.4752702473971	1.09407585874928\\
22.7812208088007	1.04181212723547\\
23.0871713702043	0.9920440771341\\
23.3931219316079	0.94465290942104\\
23.6990724930115	0.899525186102134\\
24.0703764631503	0.847642873422888\\
24.4416804332892	0.798752786690551\\
24.8129844034281	0.752682706237172\\
25.1842883735669	0.70926968674362\\
25.6498906433014	0.658351316397958\\
26.1154929130359	0.61108856627476\\
26.5810951827703	0.567219710585304\\
27.0466974525048	0.526500085729468\\
27.6555541649569	0.477626822191858\\
28.2644108774089	0.433291036452102\\
28.8732675898609	0.393073741986008\\
29.482124302313	0.356589459281667\\
30.3119878778925	0.31224287013101\\
31.1418514534721	0.273413422855851\\
31.9717150290517	0.239423393569464\\
32.8015786046313	0.209659926211921\\
33.870217925904	0.176690623725601\\
34.9388572471766	0.148909245644803\\
36.0074965684493	0.125520729157421\\
37.076135889722	0.105809371421116\\
37.8184919701046	0.0939580399779853\\
38.5608480504873	0.0834345009825241\\
39.3032041308699	0.0740913653332114\\
40.0455602112525	0.0657946184689034\\
40.7879162916351	0.0584251925243183\\
41.5302723720177	0.0518814207230347\\
42.2726284524003	0.0460716541369264\\
43.0149845327829	0.0409125546304581\\
43.9041183472765	0.0354859943334284\\
44.79325216177	0.0307795201055655\\
45.6823859762636	0.0266990432741094\\
46.5715197907572	0.0231597128067429\\
47.4286398430679	0.0201911515964861\\
48.2857598953786	0.0176032473310333\\
49.1428799476893	0.0153478623277787\\
50	0.0133815281699332\\
};
\addlegendentry{$g(t)$}

\end{axis}
\end{tikzpicture}%

%% file: 3a.tex
%
%
\definecolor{mycolor1}{rgb}{0.00000,0.44700,0.74100}%
\definecolor{mycolor2}{rgb}{0.85000,0.32500,0.09800}%
\definecolor{mycolor3}{rgb}{0.92900,0.69400,0.12500}%
\definecolor{mycolor4}{rgb}{0.49400,0.18400,0.55600}%
\definecolor{mycolor5}{rgb}{0.46600,0.67400,0.18800}%
\definecolor{mycolor6}{rgb}{0.30100,0.74500,0.93300}%
\definecolor{mycolor7}{rgb}{0.63500,0.07800,0.18400}%
\begin{tikzpicture}

\begin{axis}[%
 axis lines=left,
 x   axis line style={->},
  y   axis line style={->},
  width=\l cm,
height=\h cm,
at={(0.758in,0.481in)},
scale only axis,
xmin=0,
xmax=4.85,
ylabel={},
xlabel={time, $t$},
ymin=0,
ymax=3.45,
axis background/.style={fill=white},legend style={at={(.4,1)}, legend image post style={xscale=0.6},legend cell align=left, align=left, draw=none,fill=none},
]
\addplot [ultra thick, color=mycolor1]
  table[row sep=crcr]{%
0	0.99\\
0.167459095433972	0.99437197596897\\
0.88809699331047	0.999022174852034\\
1.48779402054587	0.998696089589295\\
1.7055453710483	0.995375074804815\\
1.92329672155073	0.940370656410951\\
2.08481058580715	0.549818034548405\\
2.15295250789078	0.303934180299088\\
2.22109442997442	0.13958178729448\\
2.28978116665248	0.0575302838472321\\
2.34992363054925	0.0254732869145372\\
2.40768040033215	0.0114766822307919\\
2.46450942917201	0.0052036243004574\\
2.52094002183088	0.00236567010425615\\
2.57719370734109	0.00107674226906496\\
2.64268384310993	0.000430529197649743\\
2.71858857476824	0.000148910719354967\\
2.80881721857623	4.23494547070469e-05\\
2.91886602012185	9.3615991872964e-06\\
3.05781743905216	1.63511196027422e-06\\
3.2410232452166	4.35052379121892e-07\\
3.49148164923853	6.28608458286273e-07\\
3.80194303114913	3.54225082255244e-06\\
4.00234327976475	1.40804550674556e-06\\
4.20274352838037	5.59692857393657e-07\\
4.45738319526598	8.96701155106427e-07\\
4.75172596488285	3.59498844426653e-06\\
4.95377003700084	1.48949153022659e-06\\
5.15581410911884	6.17125911483996e-07\\
5.40729016259262	9.1449554939782e-07\\
5.6965780011567	3.28302152174942e-06\\
5.90234655997052	1.49589984754783e-06\\
6.10811511878434	6.81594573781474e-07\\
6.35866516034022	9.8709057094748e-07\\
6.64244528416109	3.13546488206473e-06\\
6.90003906189922	5.40017854511188e-06\\
7.10814325704264	2.61347286611979e-06\\
7.30161659763872	8.76428750573588e-07\\
7.52783011201308	6.81152644177452e-07\\
7.80244400977411	1.75636308086258e-06\\
8.08159085296147	5.02494485843386e-06\\
8.31185540356826	4.34082920430789e-06\\
8.5049401740277	1.44227810039421e-06\\
8.70938519162876	6.35228217249484e-07\\
8.96202464666783	9.68805825118935e-07\\
9.24941001798386	3.3348639678632e-06\\
9.4549141215276	1.50923081876364e-06\\
9.66041822507134	6.83011469604367e-07\\
9.91057418536101	9.79500850872839e-07\\
10.1943150768653	3.10861956231978e-06\\
10.4523135903314	5.40691845463126e-06\\
10.660717391311	2.63712064263455e-06\\
10.8540987471409	8.82418803582752e-07\\
11.0799211588952	6.78830772704586e-07\\
11.3542397162077	1.73847603279764e-06\\
11.633633083963	5.00173337292242e-06\\
11.8643313735568	4.3697783478347e-06\\
12.057534127128	1.45598223377785e-06\\
12.2616888721976	6.36521893859765e-07\\
12.5139065268766	9.60712664310748e-07\\
12.8012599038722	3.30465984121897e-06\\
13.0070756278237	1.50758660744099e-06\\
13.2128913517752	6.87753094918975e-07\\
13.4630427858913	9.86189905072266e-07\\
13.7463922439436	3.10248501811383e-06\\
14.0041070819434	5.3591389961305e-06\\
14.2126343459927	2.62219288826859e-06\\
14.4063412869868	8.8427897015908e-07\\
14.6323632865856	6.83821398839157e-07\\
14.9065262749602	1.74497142335639e-06\\
15.1855399979448	4.97724742531958e-06\\
15.4161250610235	4.33560482782092e-06\\
15.6095334711894	1.4516992887792e-06\\
15.8139764999289	6.39345130948321e-07\\
16.0662875602627	9.67204883073336e-07\\
16.3533679058481	3.30691465855182e-06\\
16.5591985889874	1.50919569319452e-06\\
16.7650292721268	6.88752035182816e-07\\
17.0151245067794	9.86242814859928e-07\\
17.2984029206949	3.09770663109443e-06\\
17.5561272112081	5.35211499153032e-06\\
17.7647174056752	2.62303136866482e-06\\
17.9584659723432	8.85443758757987e-07\\
18.1844662422845	6.84333230860368e-07\\
18.4585614189364	1.74354157711498e-06\\
18.7375461972883	4.96989267684267e-06\\
18.9681735563205	4.33396015186993e-06\\
19.1616332370107	1.45292637609984e-06\\
19.3660836668957	6.40006739938137e-07\\
19.6183506758471	9.6715197550612e-07\\
19.9053805035207	3.30303029316692e-06\\
20	8.85994267618706e-07\\
};

\addplot [ultra thick,color=mycolor2]
  table[row sep=crcr]{%
0	0.2\\
0.167459095433972	0.246100068650257\\
0.88809699331047	0.673674332044964\\
1.48779402054587	1.65223176867484\\
1.7055453710483	2.23788037730758\\
1.92329672155073	2.85271642867645\\
2.08481058580715	3.12506008975829\\
2.15295250789078	3.16933825032963\\
2.22109442997442	3.18940664805245\\
2.28978116665248	3.19775221232105\\
2.34992363054925	3.20077707876863\\
2.40768040033215	3.20206020623569\\
2.46450942917201	3.20262815819542\\
2.52094002183088	3.20288367742374\\
2.57719370734109	3.20299943710217\\
2.64268384310993	3.20305740575597\\
2.71858857476824	3.20308265414876\\
2.80881721857623	3.2030922056005\\
2.91886602012185	3.20309516216335\\
3.05781743905216	3.20309585463852\\
3.2410232452166	3.20309596219139\\
3.49148164923853	3.20309594484401\\
3.80194303114913	3.20309568371598\\
4.00234327976475	3.20309587498528\\
4.20274352838037	3.20309595101671\\
4.45738319526598	3.20309592081252\\
4.75172596488285	3.20309567898432\\
4.95377003700084	3.2030958676802\\
5.15581410911884	3.2030959458636\\
5.40729016259262	3.20309591921188\\
5.6965780011567	3.20309570693735\\
5.90234655997052	3.20309586709975\\
6.10811511878434	3.20309594007945\\
6.35866516034022	3.20309591269932\\
6.64244528416109	3.20309572015497\\
6.90003906189922	3.20309551717313\\
7.10814325704264	3.20309576691105\\
7.30161659763872	3.20309592258856\\
7.52783011201308	3.20309594008923\\
7.80244400977411	3.20309584372493\\
8.08159085296147	3.2030955507831\\
8.31185540356826	3.20309561207279\\
8.5049401740277	3.20309587184442\\
8.70938519162876	3.20309594417399\\
8.96202464666783	3.20309591427716\\
9.24941001798386	3.20309570222365\\
9.4549141215276	3.20309586583745\\
9.66041822507134	3.20309593988492\\
9.91057418536101	3.20309591331199\\
10.1943150768653	3.20309572249338\\
10.4523135903314	3.20309551650175\\
10.660717391311	3.20309576472415\\
10.8540987471409	3.20309592198417\\
11.0799211588952	3.20309594022977\\
11.3542397162077	3.20309584526048\\
11.633633083963	3.20309555279591\\
11.8643313735568	3.20309560941084\\
12.057534127128	3.20309587054866\\
12.2616888721976	3.20309594399048\\
12.5139065268766	3.20309591493494\\
12.8012599038722	3.20309570486307\\
13.0070756278237	3.20309586591728\\
13.2128913517752	3.20309593939242\\
13.4630427858913	3.20309591264496\\
13.7463922439436	3.2030957229756\\
14.0041070819434	3.20309552071633\\
14.2126343459927	3.20309576599468\\
14.4063412869868	3.20309592175011\\
14.6323632865856	3.20309593971515\\
14.9065262749602	3.20309584461099\\
15.1855399979448	3.20309555492297\\
15.4161250610235	3.2030956124064\\
15.6095334711894	3.20309587086537\\
15.8139764999289	3.20309594367032\\
16.0662875602627	3.20309591428594\\
16.3533679058481	3.20309570459382\\
16.5591985889874	3.20309586570589\\
16.7650292721268	3.20309593923571\\
17.0151245067794	3.20309591257303\\
17.2984029206949	3.20309572333666\\
17.5561272112081	3.20309552127869\\
17.7647174056752	3.20309576585239\\
17.9584659723432	3.20309592157857\\
18.1844662422845	3.20309593960213\\
18.4585614189364	3.20309584467198\\
18.7375461972883	3.20309555551497\\
18.9681735563205	3.20309561248669\\
19.1616332370107	3.20309587068828\\
19.3660836668957	3.2030959435439\\
19.6183506758471	3.20309591422356\\
19.9053805035207	3.20309570487482\\
20	3.20309592149851\\
};

\end{axis}
\end{tikzpicture}%

%% file: 3b.tex
%
%
\definecolor{mycolor1}{rgb}{0.00000,0.44700,0.74100}%
\definecolor{mycolor2}{rgb}{0.85000,0.32500,0.09800}%
\definecolor{mycolor3}{rgb}{0.92900,0.69400,0.12500}%
\definecolor{mycolor4}{rgb}{0.49400,0.18400,0.55600}%
\definecolor{mycolor5}{rgb}{0.46600,0.67400,0.18800}%
\definecolor{mycolor6}{rgb}{0.30100,0.74500,0.93300}%
\definecolor{mycolor7}{rgb}{0.63500,0.07800,0.18400}%
\begin{tikzpicture}

\begin{axis}[%
 axis lines=left,
 x   axis line style={->},
  y   axis line style={->},
  width=\l cm,
height=\h cm,
at={(0.758in,0.481in)},
scale only axis,
xmin=0,
xmax=99.5,
ylabel={},
xlabel={time, $t$},
ymin=0,
ymax=1.25,
axis background/.style={fill=white},legend style={at={(.45,1)}, legend image post style={xscale=0.6},legend cell align=left, align=left, draw=none,fill=none},
]
\addplot [ultra thick, color=mycolor1]
  table[row sep=crcr]{%
0	0.99\\
0.50745180434537	0.983767087423495\\
1.01490360869074	0.973813005294851\\
1.52235541303611	0.95842552085584\\
2.02980721738148	0.935833397850007\\
2.60073903188273	0.900008049281939\\
3.17167084638398	0.852812293001088\\
3.74260266088524	0.796872686406047\\
4.31353447538649	0.738365761064979\\
4.88446628988775	0.684036680991945\\
5.455398104389	0.636342241984134\\
6.02632991889025	0.596578293876345\\
6.59726173339151	0.564559873123092\\
7.36755219755859	0.531100758215582\\
8.13784266172567	0.506887209182327\\
8.90813312589276	0.489911415889242\\
9.67842359005984	0.477385878934215\\
10.2957576065354	0.469142938231935\\
10.9130916230111	0.462482331892971\\
11.5304256394867	0.457041882810835\\
12.1477596559623	0.452463827585113\\
12.9200889238252	0.447537444597003\\
13.6924181916881	0.443318267119472\\
14.464747459551	0.439614950427553\\
15.237076727414	0.43620441461014\\
16.1967856470289	0.432136514906222\\
17.1564945666439	0.428218349792036\\
18.1162034862589	0.424374045455183\\
19.0759124058738	0.42048885402433\\
20.3508892097396	0.415140994818342\\
21.6258660136054	0.409580935359147\\
22.9008428174712	0.40376739798669\\
24.1758196213369	0.397646534775049\\
26.233236809528	0.387226607547933\\
28.2906539977191	0.375763566003221\\
30.3480711859102	0.362967875613061\\
32.4054883741013	0.349207176245444\\
34.2427237574628	0.336390177691429\\
36.0799591408243	0.322651860079443\\
37.9171945241858	0.307911106865341\\
39.7544299075473	0.292404385337318\\
41.5240440270539	0.276963489006033\\
43.2936581465605	0.260937584500663\\
45.0632722660671	0.244398732069434\\
46.8328863855737	0.22753140076071\\
48.8023319451375	0.208626047149765\\
50.7717775047013	0.189760702897144\\
52.7412230642651	0.171163746884646\\
54.7106686238289	0.153087151754547\\
57.0067208712743	0.133014797824435\\
59.3027731187197	0.114321081935334\\
61.5988253661651	0.0972306012075929\\
63.8948776136106	0.0818856712137009\\
66.3948776136106	0.0672387496856876\\
68.8948776136106	0.0546960118834564\\
71.3948776136106	0.0441321550829552\\
73.8948776136106	0.035372331889374\\
76.3948776136106	0.0281867998937796\\
78.8948776136106	0.0223516005296121\\
81.3948776136106	0.0176639629907999\\
83.8948776136106	0.0139218622422662\\
86.1398756454033	0.0112121813038603\\
88.384873677196	0.00901501623749894\\
90.6298717089887	0.00724335150678087\\
92.8748697407814	0.00581504827505591\\
94.6561523055861	0.00487954639702222\\
96.4374348703907	0.00409273896686197\\
98.2187174351953	0.00343219991234664\\
100	0.00287747632991752\\
};

\addplot [ultra thick,color=mycolor2]
  table[row sep=crcr]{%
0	0.1\\
0.50745180434537	0.105136546254987\\
1.01490360869074	0.110492835854544\\
1.52235541303611	0.116048156229475\\
2.02980721738148	0.121765784930381\\
2.60073903188273	0.128330964446524\\
3.17167084638398	0.134922265086982\\
3.74260266088524	0.141425383492721\\
4.31353447538649	0.147758447287794\\
4.88446628988775	0.153875741217139\\
5.455398104389	0.159780929068128\\
6.02632991889025	0.16550003177576\\
6.59726173339151	0.171071738955016\\
7.36755219755859	0.178450408998194\\
8.13784266172567	0.185735514782985\\
8.90813312589276	0.19298766477381\\
9.67842359005984	0.200292152485228\\
10.2957576065354	0.206230045432518\\
10.9130916230111	0.212246354419209\\
11.5304256394867	0.218354316335292\\
12.1477596559623	0.224568593904208\\
12.9200889238252	0.23251157448131\\
13.6924181916881	0.240650326527273\\
14.464747459551	0.248994610626255\\
15.237076727414	0.257557287869329\\
16.1967856470289	0.26851738713547\\
17.1564945666439	0.279836372557836\\
18.1162034862589	0.291520816158559\\
19.0759124058738	0.30358026990644\\
20.3508892097396	0.320192278061061\\
21.6258660136054	0.337478525781776\\
22.9008428174712	0.3554385845297\\
24.1758196213369	0.374069611184692\\
26.233236809528	0.405495044182962\\
28.2906539977191	0.438583198564751\\
30.3480711859102	0.473253066254292\\
32.4054883741013	0.509277664984683\\
34.2427237574628	0.542361900195317\\
36.0799591408243	0.576195686956247\\
37.9171945241858	0.610590703351707\\
39.7544299075473	0.645249283194956\\
41.5240440270539	0.678577672944233\\
43.2936581465605	0.711646248830635\\
45.0632722660671	0.744207303069646\\
46.8328863855737	0.775964165967156\\
48.8023319451375	0.810000497552903\\
50.7717775047013	0.842391740576373\\
52.7412230642651	0.872880128501536\\
54.7106686238289	0.901212569172203\\
57.0067208712743	0.931262168408858\\
59.3027731187197	0.958043651641663\\
61.5988253661651	0.981581456057937\\
63.8948776136106	1.00195814979953\\
66.3948776136106	1.02076439660849\\
68.8948776136106	1.03640185173152\\
71.3948776136106	1.04924493958329\\
73.8948776136106	1.05968387211468\\
76.3948776136106	1.06811703812219\\
78.8948776136106	1.07486986482941\\
81.3948776136106	1.08022856505044\\
83.8948776136106	1.08447082169318\\
86.1398756454033	1.08752817624016\\
88.384873677196	1.08999446339286\\
90.6298717089887	1.09197428188964\\
92.8748697407814	1.0935653879761\\
94.6561523055861	1.09460555971726\\
96.4374348703907	1.09547886130882\\
98.2187174351953	1.09621091271838\\
100	1.09682495424048\\
};

\end{axis}
\end{tikzpicture}%